\documentclass[10pt]{amsproc}
\usepackage{amssymb}
\usepackage{color}
\usepackage{graphicx,color}
\usepackage{amsmath}
\usepackage{amscd} 

\theoremstyle{plain}
 \newtheorem{thm}{Theorem}[section]
 \newtheorem{prop}{Proposition}[section]
 \newtheorem{lem}{Lemma}[section]
 \newtheorem{cor}{Corollary}[section]
\theoremstyle{definition}
 
 \newtheorem{dfn}{Definition}[section]
\theoremstyle{remark}
 \newtheorem{rem}{Remark}[section] 
 \numberwithin{equation}{section}

\renewcommand{\leq}{\leqslant}
\renewcommand{\geq}{\geqslant}

\title[Quantum and Spectral Properties of the Labyrinth Model]{QUANTUM AND SPECTRAL PROPERTIES OF THE LABYRINTH MODEL}

\author[Y.\ Takahashi]{\bfseries YUKI TAKAHASHI}

\address{Department of Mathematics, University of California, Irvine, CA~92697, USA}

\email{takahasy@math.uci.edu}

\thanks{Y.\ T. \ was supported in part by NSF grant DMS-1301515 (PI: A.\ Gorodetski).}

\date{today}

\begin{document}

\vspace{18mm}
\setcounter{page}{1}
\thispagestyle{empty}

\begin{abstract}
We consider the Labyrinth model, which is a two-dimensional quasicrystal model. We show that the spectrum of this model, which is known to be a product of two Cantor sets, is an interval for small values of the coupling constant. We also consider the density of states measure of the Labyrinth model, and show that it is absolutely continuous with respect to Lebesgue measure for almost all values of coupling constants in the small coupling regime.
\end{abstract}

\maketitle

\section{Introduction}  
\subsection{Quasicrystal and the Labyrinth model.}
The Fibonacci Hamiltonian is a central model in the study of electronic properties of one-dimensional quasicrystals.
It is given by the following bounded self-adjoint operator in $l^{2}({\mathbb{Z})}$:
\begin{equation}\label{FH}
(H_{\lambda, \beta} \psi)(n) = \psi(n+1) + \psi(n-1) + \lambda \chi_{[1-\alpha, 1)} ( n\alpha + \beta \mod 1 ) \psi (n),
\end{equation}
where $\alpha = \frac{ \sqrt{5} - 1}{2}$ is the \emph{frequency}, 
$\beta \in \mathbb{T} = \mathbb{R}/\mathbb{Z}$ is the \emph{phase}, and $\lambda > 0$ is the \emph{coupling constant}. 
By the minimality of the circle rotation and strong operator convergence, 
the spectrum is easily seen to be independent of $\beta$. 
With this specific choice of $\alpha$, when $\beta = 0$ the potential of (\ref{FH}) coincides with the 
\emph{Fibonacci substitution sequence} (for the precise definition, see section \ref{preliminaries}).
Papers on this model include \cite{DGET}, \cite{DG0902}, \cite{DG11}, \cite{DG12}, \cite{Fibonacci}. 
In \cite{DG11}, the authors showed that for sufficiently small coupling constant, 
the spectrum is a dynamically defined Cantor set, and the \emph{density of states measure} is exact dimensional. 
Later, this result was extended for all values of the coupling constant \cite{Fibonacci}. 

In physics papers, it is more traditional to consider off-diagonal model, but in fact they are known to be very similar.
The operator of the off-diagonal model has the following form:
\begin{equation}\label{ST}
(H_{\omega}\psi)(n) = \omega(n+1) \psi(n+1) + \omega(n) \psi(n-1),
\end{equation}
where the sequence $\omega$ is in the \emph{hull} of the Fibonacci substitution sequence. 
For the precise definition of hull, see (\ref{hull}). 
This sequence takes two positive real values, say $1$ and $a$. 
Let 
\begin{equation}\label{CC}
\lambda = \frac{| a^2 - 1 |}{a}, 
\end{equation}
and call this the \emph{coupling constant}. 
The spectral properties of $H_{\omega}$ do not depend on the particular choice of $\omega$, 
and depend only on the coupling constant $\lambda$. 
Recent mathematics papers discussing this operator include \cite{DG11}, \cite{WM}, \cite{Will}.
In \cite{WM} the authors considered tridiagonal substitution Hamiltonians, 
which include both (\ref{FH}) and (\ref{ST}) as special cases.
\begin{centering}
\begin{figure}[t]
\includegraphics[scale=1.00]{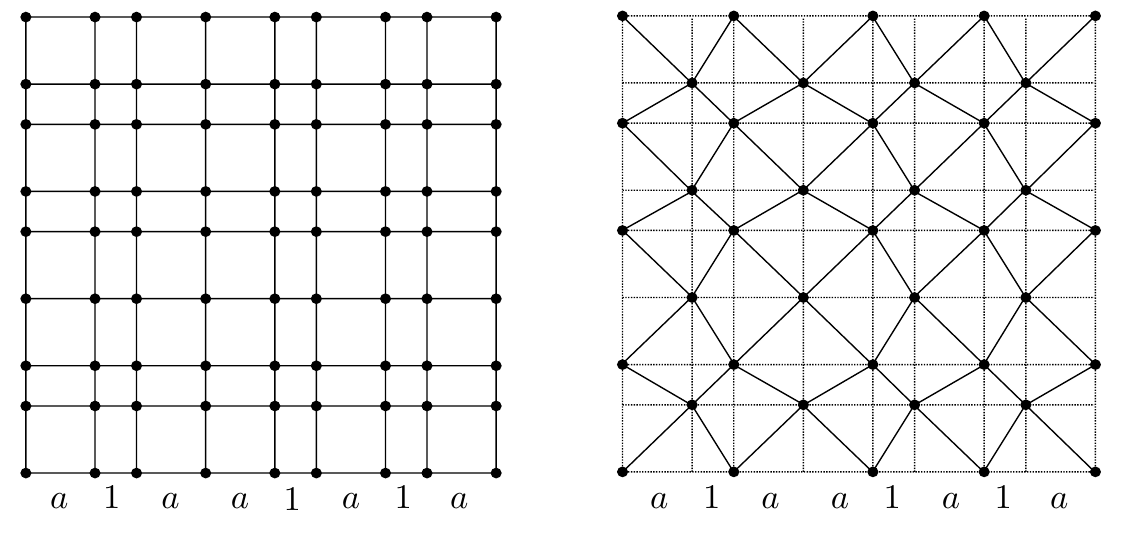}
\caption{The square tiling (left) and the Labyrinth model (right).}
\label{squaretiling}
\end{figure}
\end{centering}

It is natural to consider higher dimensional models, but that is known to be extremely difficult.
To get an idea of spectral properties of higher dimensional quasicrystals, simpler models have been considered.
In two dimensional case, we have, for example, the square Fibonacci Hamiltonian \cite{DG13},  
the square tiling, and 
the Labyrinth model. The Labyrinth model is the main subject of this paper.
These models are separable, so the existing results of one-dimensional models can be applied in the study of their spectral properties.

The square Fibonacci Hamiltonian is constructed by two copies of the Fibonacci Hamiltonian. Namely, 
this operator acts on $l^{2}( \mathbb{Z}^{2} )$, and is given by 
\begin{equation*}
\begin{aligned}
\left[ H_{\lambda_{1}, \lambda_{2}, \beta_{1}, \beta_{2}} \psi \right] (m, n) = 
\psi(m+1, n) + \psi(m-1, n) + \psi(m, n+1) + \psi(m, n-1)  \\
\hspace{1mm} + \left( \lambda_{1} \chi_{[1-\alpha, 1)} ( m \alpha + \beta_{1} \mod 1) + 
\lambda_{2}\chi_{[1-\alpha, 1)}( n \alpha + \beta_{2} \mod 1 ) \right) \psi(m, n),
\end{aligned}
\end{equation*}
where $\alpha = \frac{\sqrt{5} - 1}{2}$, $\beta_{1}, \beta_{2} \in \mathbb{T}$, and $\lambda_{1}, \lambda_{2} > 0$.
It is known that the spectrum of this operator is given by the sum of the spectra of the one-dimensional models, 
and the density of states measure of this operator is 
the convolution of the density of states measures of the one-dimensional models. See, for example, the appendix in \cite{DG13}.
Recently, it was shown that for small coupling constants the spectrum of the square Fibonacci Hamiltonian is an interval \cite{DG11}.
Furthermore, it was shown that for almost all pairs of the coupling constants, 
the density of states measure is absolutely continuous with respect to Lebesgue measure in weakly coupled regime \cite{DG13}.

The square tiling is constructed by two copies of off-diagonal models. The operator acts on $l^{2}(\mathbb{Z}^2)$, and is given by 
\begin{equation*}
\begin{aligned}
\left[H_{\omega_{1}, \omega_{2}}\psi \right] (m, n) = \omega_{1}(m+1) &\psi(m+1, n) + \omega_{1}(m) \psi(m-1, n)  \\
&\hspace{2mm} + \omega_{2}(n+1) \psi(m, n+1) + \omega_{2}(n) \psi(m, n-1),
\end{aligned}
\end{equation*}
where the sequences $\omega_{1}$ and $\omega_{2}$ are in the hull of the Fibonacci substitution sequence. 
All vertices are connected horizontally and vertically. 
See Figure \ref{squaretiling}.
It has been mainly studied numerically by physicists (e.g., \cite{EL}, \cite{ILED}, \cite{Lifshitz}).  
By repeating the argument from \cite{DG13}, 
one can show that the analogous results of the square Fibonacci Hamiltonian hold for the 
square tiling. Recently, \cite{FTY} considered the square tridiagonal Fibonacci Hamiltonians, which 
include the square Fibonacci Hamiltonian and the square tiling as special cases. 
 
 The operator of the Labyrinth model is given by:
\begin{equation}\label{Labyrinth}
\begin{aligned}
\left[ \hat{H}_{\omega_{1}, \omega_{2}} \psi \right] (m, n) &=  \omega_{1}(m+1) \omega_{2}(n+1) \psi( m+1, n+1 ) \\
&\hspace{10mm}+\omega_{1}(m+1) \omega_{2}(n) \psi( m+1, n-1 ) \\
&\hspace{14mm}+ \omega_{1}(m) \omega_{2}(n+1) \psi( m-1, n+1 ) \\
&\hspace{24mm} + \omega_{1}(m) \omega_{2}(n) \psi( m-1, n-1 ),
\end{aligned}
 \end{equation}
 where the sequences $\omega_{1}$ and $\omega_{2}$ are in the hull of the Fibonacci substitution sequence. 
 It is constructed by two copies of off-diagonal models.
 All vertices are connected diagonally, and the strength 
of the bond is equal to the product of the sides of the rectangle. See Figure \ref{squaretiling}. 
Compare with Figure 2 from \cite{Sire}. 
Without loss of generality, we can assume that $\omega_1$ and $\omega_2$ take values in $\{1, a_1\}$ and $\{1, a_2\}$, respectively. 
We denote the corresponding coupling constants by $\lambda_1$ and $\lambda_2$. 
It can be shown that the spectral properties do not depend on the specific choice of $\omega_{1}$ and $\omega_{2}$, and only depend on 
the coupling constants $\lambda_1$ and $\lambda_2$.  
Unlike the square Fibonacci Hamiltonian or the square tiling, the spectrum is the product (not the sum) of the spectra of the two one-dimensional models, and the density of states measure is not the convolution of the density of states measures of the one-dimensional models. 
This model was suggested in the late 1980s in \cite{Sire}, and 
so far this has been studied mostly by physicists, and their work is mainly relied on numerics \cite{CM}, \cite{RTS}, \cite{Sire}, \cite{Sire2}, \cite{TS}, \cite{TS1}, \cite{TS2}, \cite{TS3}, \cite{TS4}, \cite{YGRS}.
Sire considered this model in \cite{Sire} and 
the numerical experiments suggested that 
the density of states measure is absolutely continuous for small coupling constants 
and singular continuous for large coupling constants.
By a heuristic argument, the author also estimated the critical value of which the transition from zero measure spectrum to positive measure spectrum occurs, and showed that it agrees with numerical experiment. 
In some papers, other substitution sequences, e.g., silver mean sequence or bronze mean sequence, are 
 used to define the Labyrinth model. We consider more general cases in this paper, 
 and give rigorous proofs to some of the physicists' conjectures. 
In physicists' work, the coupling constants of two substitution sequences $\omega_{1}$ and $\omega_{2}$
are set as equal, 
but we consider the case that they may be different. 
We denote the spectrum of (\ref{Labyrinth}) by $\hat{\Sigma}_{\lambda_1, \lambda_2}$, and 
the density of states measures of (\ref{ST}) and (\ref{Labyrinth}) by  
$\nu_{\lambda}$ and $\hat{\nu}_{\lambda_{1}, \lambda_{2}}$, respectively. 
The following theorems are the main results of this paper.
\begin{thm}\label{mainnnnntheorem}
The spectrum $\hat{\Sigma}_{\lambda_1, \lambda_2}$
is a Cantor set of zero Lebesgue measure for sufficiently large coupling constants and 
is an interval for sufficiently small coupling constants. 
\end{thm}

\begin{thm}\label{maintheorem2}
For any $E \in \mathbb{R}$, 
\begin{equation}\label{density1}
\hat{\nu}_{\lambda_{1}, \lambda_{2}} \left( (-\infty, E] \right) = 
\iint_{\mathbb{R}^{2} } \chi_{(-\infty, E]}(xy) \, d\nu_{\lambda_{1}} (x)d\nu_{\lambda_{2}}(y).
\end{equation}
The density of states measure $\hat{\nu}_{\lambda_{1}, \lambda_{2} }$ 
is singular continuous for sufficiently large coupling constants. 
Furthermore, there exists $\lambda^{\ast} > 0$ such that for almost every pair 
$(\lambda_{1}, \lambda_{2}) \in [0, \lambda^{\ast}) \times [0, \lambda^{\ast})$, the density of states measure 
$\hat{\nu}_{\lambda_{1}, \lambda_{2} }$ is absolutely continuous with respect to 
Lebesgue measure.   
\end{thm}

\subsection{Structure of the paper}
In section 2, we introduce metallic mean sequences and prove some lemmas. 
We then define off-diagonal model and discuss necessary results. In section 3 we define the Labyrinth model, and using the results 
in section 2, we prove Theorem 1.1 and 1.2.

\section{Preliminaries}\label{preliminaries}

\subsection{Linearly recurrent sequences}
We recall some basic facts about subshifts over two symbols. 

An \emph{alphabet} is a finite set of symbols called \emph{letters}.  
A \emph{word} on $\mathcal{A}$ is a finite nonempty sequence of letters. Write $\mathcal{A}^{+}$ for the set of words. 
For $u = u_{1} u_{2} \cdots u_{n} \in \mathcal{A}^{+}$, $|u| = n$ is the \emph{length} of $u$. 
Define the \emph{shift} $T$ on $\mathcal{A}^{\mathbb{Z}}$ by 
\begin{equation*}
(Tx)_{n} = x_{n+1}. 
\end{equation*}
Assume that $\mathcal{A}^{\mathbb{Z}}$ is equipped with the product topology. 
A \emph{subshift} $(X, T)$ on an alphabet $\mathcal{A}$ is a closed 
$T$-invariant subset $X$ of $\mathcal{A}^{\mathbb{Z}}$, endowed with 
the restriction of $T$ to $X$, which we denote again by $T$. 
Given $u = u_1 u_2 \cdots u_n \in \mathcal{A}^{+}$ and an interval $J = \{i, \cdots, j\} \subset \{1, 2, \cdots, n \}$, 
we write $u_J$ to denote the word $u_i u_{i+1} \cdots u_{j}$. A \emph{factor} of 
$u$ is a word $v$ such that $v = u_{J}$ for some interval $J \subset \{ 1, 2, \cdots, n \}$. 
We extend this definition in obvious way to $u \in \mathcal{A}^{\mathbb{Z}}$.  
The \emph{language} $\mathcal{L}(X)$ of a subshift 
$(X, T)$ is the set of all words that are factors of at least one element of $X$.
\begin{dfn}
Let $(X, T)$ be a subshift. We say that $x \in X$ is \emph{linearly recurrent} if there exists a constant $K > 0$ such that 
for every factor $u, v$ of $x$, $K |u| < |v|$ implies that $u$ is a factor of $v$.  
\end{dfn}
We say that a subshift is \emph{linearly recurrent} if it is minimal and contains a linearly recurrent sequence. 
Note that if a subshift is linearly recurrent, then by minimality, all sequences belonging to $X$ are linearly recurrent. 

\subsection{Metallic mean sequence }
Let $\mathcal{A} = \{ a, b \}$ be an alphabet, and consider the following substitution:
\begin{equation*}
\mathcal{P}_{s} : 
\begin{cases}
a \longrightarrow a^{s}b \\
b \longrightarrow a,
\end{cases}
\end{equation*}
where $s$ is a positive integer. 
Consider the iteration of $\mathcal{P}_{s}$ on $a$. 
For example, if $s=1$, 
\begin{equation*}
a \longrightarrow ab \longrightarrow aba \longrightarrow abaab \longrightarrow  abaababa \longrightarrow \cdots.
\end{equation*}
Let us write the $n$th iteration as $\mathcal{C}_s(n)$. 
It is easy to see that  
\begin{equation*}
\mathcal{C}_{s}(n+1) = ( \mathcal{C}_{s}(n) )^{s} \mathcal{C}_{s}(n - 1).
\end{equation*} 
Therefore, for any $s \in \mathbb{N}$ we can define a sequence $\{ u_{s}(k) \}_{k=1}^{\infty}$ by 
$u_{s} = \lim_{n \to \infty} \mathcal{C}_{s}(n)$. 
They are called \emph{metallic mean sequences}. In particular, 
when $s=1$, it is called the \emph{Fibonacci substitution sequence} or \emph{golden mean sequence}. When $s= 2, 3$, 
they are called the \emph{silver mean sequence} and \emph{bronze mean sequence}, respectively.

We define the \emph{hull} $\Omega^{(s)}_{a, b}$ of $u_{s}$ by 
\begin{equation}\label{hull}
\Omega^{(s)}_{a, b} = \left\{  \omega \in \{ a, b\}^{\mathbb{Z}} \mid 
\text{ every factor of $\omega$ is a factor of $u_{s}$}  \right\}.
\end{equation}
It is well known that $\Omega^{(s)}_{a, b}$ is compact and $T$-invariant and $(\Omega^{(s)}_{a, b}, T)$ is linearly recurrent. 
See for example, \cite{Martine} and references therein.

\begin{rem}

Let us define a \emph{rotation sequence} $v_{a, b, s, \beta}$ by 
\begin{equation*}\label{rotationsequence}
v_{a, b, s, \beta}(n) = 
\begin{cases}
a & \text{ \ if \ } n \alpha + \beta \mod 1  \in [1 - \alpha, 1) \\
b & \text{ \ o.w.}, 
\end{cases}
\end{equation*}
where $\alpha$ is given by 
\begin{equation*}\label{alpha1}
\begin{aligned}
\alpha = \cfrac{1}{ 1 + s + \cfrac{1}{s + \cfrac{1}{s + \cfrac{1}{\ddots}  } } } = \frac{ s + 2 - \sqrt{s^2 + 4} }{2s}. 
\end{aligned}
\end{equation*}
It is easy to see that the potential of the Fibonacci Hamiltonian (\ref{FH}) is $v_{\lambda, 0, 1, \beta}$. 
It is well known that $v_{a, b, s, 0} = u_s$, so there is no need to distinguish the rotation sequence and substitution sequence. 
However, it seems that it is more common to use the rotation sequence in the definition of the 
on-diagonal model and use the substitution sequence 
in the definition of the off-diagonal model. 
It is also known that 
\begin{equation*}
\Omega^{(s)}_{a, b} = \bigcup_{\beta \in \mathbb{T}} v_{a, b, s, \beta}.  
\end{equation*}
See, for example \cite{M}. 
\end{rem}

\begin{rem}\label{rem1}
There seems to be a minor confusion about substitution sequences and rotation sequences in some papers. 
Let 
\begin{equation*}\label{alpha*}
\alpha^{*} = \cfrac{1}{ s + \cfrac{1}{s + \cfrac{1}{s + \cfrac{1}{\ddots}  } } }. 
\end{equation*}
Using $\alpha^{*}$, define $v^{*}_{a, b, s, \beta}$ analogously. In some papers 
it is stated that $v^{*}_{a, b, s, 0} \in \Omega^{(s)}_{a, b}$, 
but this is obviously not true. 
What is true is that $v^{*}_{a, b, 1, 0} = v_{b, a, 1, 0}$, so when $s = 1$ there is no actual harm. 
\end{rem}

We simply write $\Omega^{(s)}_{a, b}$ as $\Omega^{(s)}$ below when there is no chance of confusion.

\subsection{Necessary results}
We will need the following definition and subsequent lemmas later.
\begin{dfn}
Let $(X, T)$ be a linearly recurrent subshift, and let $x, y \in \mathcal{L}(X)$. If  
there exist disjoint intervals $J_{1}$ and $J_{2}$ such that 
\begin{equation*}
\begin{aligned}
1) & \ J_{i} \subset \{ 1, 2, \cdots, |x|\} \ \text{ for } i= 1, 2, \\
2) & \ J_{1} = J_{2} + k \ \text{  for some odd number } k, \text{ and}  \\  
3) & \ x_{J_{1}} = x_{J_{2}} = y, 
\end{aligned}
\end{equation*}
we say that $y$ is \emph{odd-twin} in $x$. Define \emph{even-twin} analogously.  
For example, in the case of the subshift $\left( \Omega^{(1)}, T \right)$, $ab$ is odd-twin in $abaab$, 
and even-twin in $abab$.
\end{dfn}

\begin{lem}\label{twinlemma}
For any $k \geq 1$, there exists $x \in \mathcal{L}(\Omega^{(s)})$ such that $| x | \leq 3 | \mathcal{C}_{s}(k) |$ and 
$\mathcal{C}_{s}(k)$ is odd-twin in $x$.
\end{lem}
\begin{proof}
In the proof below, we simply write $\mathcal{C}_{s}(n)$ as $\mathcal{C}(n)$.
Recall that $\mathcal{C}(n)$ satisfies the concatenation rule 
\begin{equation*}
\mathcal{C}(n+1) = \mathcal{C}(n) ^{s} \mathcal{C}(n - 1). 
\end{equation*} 
Therefore, it is easy to see that $\mathcal{C}(k) \mathcal{C}(k)$ and 
$\mathcal{C}(k) \mathcal{C}(k-1) \mathcal{C}(k)$ are both factors of $\mathcal{C}(k+3)$.
If $| \mathcal{C}(k) |$ is odd, $x = \mathcal{C} (k) \mathcal{C}(k)$ satisfies the desired properties. 
Suppose $| \mathcal{C}(k) |$ is even. Note that the sequence 
\begin{equation*}
\left\{ | \mathcal{C}(n) | \mod 2 \right\}
\end{equation*}
repeats $1, 1, 1, 1, \cdots $ if $s$ 
is even, and $1, 0, 1, 1, 0, 1 \cdots$ if $s$ is odd. 
Since $| \mathcal{C}(k) |$ is even, $s$ has to be odd.
Therefore $| \mathcal{C}(k-1) |$ is odd, so 
$x = \mathcal{C}(k) \mathcal{C}(k-1) \mathcal{C}(k)$ satisfies the desired properties.
\end{proof}

\begin{lem}\label{twin}
For every $s \in \mathbb{N}$, there exists a constant $K_{s} > 0$ such that 
for any $x, y \in \mathcal{L}(\Omega^{(s)})$, $K_{s} |y| < |x|$ implies $y$ is odd-twin in $x$. 
Analogous results hold for even-twins. 
\end{lem}

\begin{proof}
Let us show the statement for odd-twin. The latter statement is immediate. 
Let $K>0$ be a number such that for any $x, y \in \mathcal{L}(\Omega^{(s)})$, 
$y$ is a factor of $x$ whenever $K |y| < | x |$.
Let $y \in \mathcal{L}(\Omega^{(s)})$.
Take $k > 0$ such that
\begin{equation*}
 | \mathcal{C}_{s}(k-1) | \leq  K |y|  < | \mathcal{C}_{s}(k) |.
 \end{equation*}
Then $y$ is a factor of $\mathcal{C}_{s}(k)$, and 
since $ | \mathcal{C}_{s}(k) | < ( s+1 ) | \mathcal{C}_{s}(k-1) | $, we have 
$|\mathcal{C}_{s}(k)| < (s+1) K |y|$. Therefore, by Lemma \ref{twinlemma}, 
there exists $v \in \mathcal{L}(\Omega^{(s)})$ such that 
$|v| < 3 (s+1) K |y|$ and $\mathcal{C}_{s}(k)$ is odd-twin in $v$. 
Since $y$ is a factor of $\mathcal{C}_{s}(k)$, $y$ is odd-twin in $v$.
Therefore,
\begin{equation*}
K_{s} := K \cdot 3(s+1) K
= 3(s+1) K^{2}
\end{equation*}
satisfies the desired properties.
\end{proof}

\subsection{The off-diagonal model}
Let $a, b > 0$ be real numbers, and let $s$ be a positive integer. 
Let $\omega \in \Omega^{(s)}_{a, b}$. We define a Jacobi matrix $H_{\omega}$ acting on $l^{2}(\mathbb{Z})$ by 
\begin{equation*}
\left( H_{\omega} \psi \right) (n) = \omega(n+1) \psi(n+1) + \omega(n) \psi(n-1),
\end{equation*}  
and set 
\begin{equation*}
\lambda = \left| \frac{  a^{2} - b^{2}  }{ ab } \right|.
\end{equation*}
We call this $\lambda$ the \emph{coupling constant}. 
We only consider the case that $a > b$. The argument is completely analogous in the case $a < b$. 
By appropriate scaling, we can always assume $b = 1$. 
We assume this scaling all throughout this section. 
Note that this coincides with the definition (\ref{CC}). See also Remark \ref{reason} below. 
We call this family of self-adjoint operators $\left\{ H_{\omega} \right\}$ 
the \emph{off-diagonal model}. By a well known argument (minimality of the subshift and strong operator convergence), 
one can see that the spectrum of $H_{\omega}$ is independent of the specific choice of $\omega$  
and depends only on $\lambda \text{ and } s$.

\begin{dfn}
We define the \emph{trace map} $T_{s}$ by 
\begin{equation*}
T_{s} = U^{s} \circ P,
\end{equation*}
 where 
\begin{equation*}
U
\begin{pmatrix}
x \\
y \\
z
\end{pmatrix}
 = 
\begin{pmatrix}
2xz - y \\
x \\
z 
\end{pmatrix}
\text{ \ and \ }
P 
\begin{pmatrix}
x \\
y \\
z
\end{pmatrix}
= 
\begin{pmatrix}
x \\
z \\
y
\end{pmatrix}.
\end{equation*}
\end{dfn}

Let $\ell_{\lambda}$ be the line given by
\begin{equation*}
\ell_{\lambda} = \left\{ \left( \frac{E^2 - a^2 - 1}{2a}, \frac{E}{2a}, \frac{E}{2}  \right) : E \in \mathbb{R} \right\}, 
\end{equation*}
and call this the \emph{line of initial condition}. 
We define the map $J_{\lambda}(\cdot)$ by 
\begin{equation}\label{identification}
J_{\lambda} : E \mapsto \ell_{\lambda}(E).
\end{equation}
The function 
\begin{equation*}
G(x, y, z) = x^{2} + y^{2} + z^{2} - 2xyz - 1
\end{equation*}
is invariant under the action of $T_{s}$ and hence preserves the family of surfaces
\begin{equation*}
S_{V} = \left\{ (x, y, z) \in \mathbb{R}^{3} : x^{2} + y^{2} + z^{2} - 2xyz - 1 =  
\frac{ V^{2} }{4} \right\}.
\end{equation*}
It is easy to see that $\ell_{\lambda} \subset S_{\lambda}$. 

The following can be proven by repeating the argument of \cite{WM}. 

\begin{thm}\label{bounded}
We have 
\begin{equation*}
\sigma(H_{\omega}) = \left\{ E \in \mathbb{R} : \text{ the forward semi-orbit of } 
J_{\lambda}(E) \text{ is bounded\,} \right\}.
\end{equation*}
\end{thm}

Notice that it is clear from this theorem that the spectrum of $H_{\omega}$ depends only on $\lambda$ and $s$. 
We denote it by $\Sigma_\lambda$ below. 

\begin{rem}\label{reason}
Let us define $\ell'_\lambda$ by  
\begin{equation*}
\ell'_{\lambda} = \left\{ \left( \frac{E^2 - \lambda E - 2}{2}, \frac{E - \lambda}{2}, \frac{E}{2}  \right) : E \in \mathbb{R} \right\}, 
\end{equation*}
and the map $J'_{\lambda}(\cdot)$ by 
\begin{equation*}
J'_{\lambda} : E \mapsto \ell'_{\lambda}(E).
\end{equation*}
It is easy to see that $\ell'_0 = \ell_0$ and $\ell'_{\lambda} \subset S_{\lambda}$. Then, the exact same statement in Theorem \ref{bounded} 
holds with $H_{\lambda, \beta}$ and $J'$ instead of $H_\omega$ and $J$, respectively. See, for example, \cite{DG11}. 
This is why we defined the coupling constant of off-diagonal model by (\ref{CC}). 
\end{rem}

By Theorem \ref{bounded}, we immediately get the following:
\begin{cor}
The spectrum $\Sigma_\lambda$ contains $0$. 
\end{cor}

\begin{proof}
Note that 
\begin{equation}
J_{\lambda}(0) = \left( -\frac{a^2 + 1}{2a}, 0, 0 \right).
\end{equation}
It is easy to see that this point is periodic under the action of $T_s$. 
\end{proof}

In what follows we are going to use some notations and results from the theory of hyperbolic dynamical systems, see \cite{KT} for some 
background on this subject. 

Let us denote by $\Lambda_{\lambda}$ the set of points whose orbits are bounded under $T_{s}$. 
The following theorem was first proven 
in \cite{S87} for the Fibonacci Hamiltonian. 

\begin{thm}[\cite{Can09}, see also Theorem 4.1 from \cite{M}]
The set $\Lambda_{\lambda}$ 
 is a compact locally maximal $T_{s}$-invariant transitive hyperbolic subset of $S_{\lambda}$, and 
 the periodic points of $T_{s}$ form a dense subset of $\Lambda_{\lambda}$.
 \end{thm}

We also have the following: 

\begin{thm}[Corollary 2.5 of \cite{DPW}]\label{thm2.3}
The forward semi-orbit of a point $p \in S_{\lambda}$ is bounded if and only if $p$ 
lies in the stable lamination of $\Lambda_{\lambda}$.
\end{thm}

The following theorem was proven in \cite{Fibonacci} for the Fibonacci Hamiltonian case, 
and recently it was extended to tridiagonal Fibonacci Hamiltonians in \cite{FTY}. It follows by repeating the argument of \cite{Fibonacci}. 

\begin{thm}\label{transverse}
For all $\lambda > 0$, the intersections of the curve of initial condition 
$\ell_{\lambda}$ with the stable lamination is transverse.
\end{thm}

\begin{cor}
The spectrum $\Sigma_\lambda$ is a dynamically defined Cantor set.
\end{cor}

Now we define the density of states measure. The definition is analogous for higher dimensional models. 

\begin{dfn}\label{density}
Denote by $H^{(N)}_{\omega}$ 
the restriction of $H_{\omega}$ to the interval $[0, N-1]$ with Dirichlet boundary conditions. 
The density of states measure $\nu_{\lambda}$ of $H_{\omega}$ is given by
\begin{equation*}
\nu_{\lambda}\left( (-\infty, E] \right) = 
\lim_{N \to \infty} \frac{1}{N} \# \left\{  \text{eigenvalues of $H^{(N)}_{\omega}$ that are in $(-\infty, E ]$} \right\}, 
\end{equation*}
where $E \in \mathbb{R}$. 
\end{dfn}
The limit does not depend on the specific choice of $\omega$, and depends only on $\lambda$ and $s$. 
In fact, the convergence is uniform in $\omega$.
This was shown in a more general setting \cite{LS2}. 

It is well known that $\Sigma_0 = [-2, 2]$, and 
\begin{equation}\label{freedensity}
\nu_{0} \left( (-\infty, E ] \right) = 
\begin{cases}
0   & E \leq -2 \\  
\frac{1}{\pi} \arccos (-\frac{E}{2}) & -2 < E < 2 \\
1 & E \geq 2.
\end{cases}
\end{equation}

Let us write
\begin{equation*}
\mathbb{S} = S_{0} \cap \left\{ (x, y, z) \in \mathbb{R}^{3} \mid |x| \leq 1, |y| \leq 1, |z| \leq 1 \right\}.
\end{equation*}
The trace map $T_{s}$ restricted to $\mathbb{S}$ is a factor of the hyperbolic automorphism $\mathcal{A}$ of 
$\mathbb{T} = \mathbb{R}^{2} / \mathbb{Z}^{2}$ given by 
\begin{equation*}
\mathcal{A} :
\begin{pmatrix}
\theta \\
\varphi \\
\end{pmatrix}
\mapsto 
\begin{pmatrix}
s & 1 \\
1 & 0 \\
\end{pmatrix}
\begin{pmatrix}
\theta \\
\varphi \\
\end{pmatrix}.
\end{equation*}

\begin{centering}
\begin{figure}[t]
\includegraphics[scale=1.00]{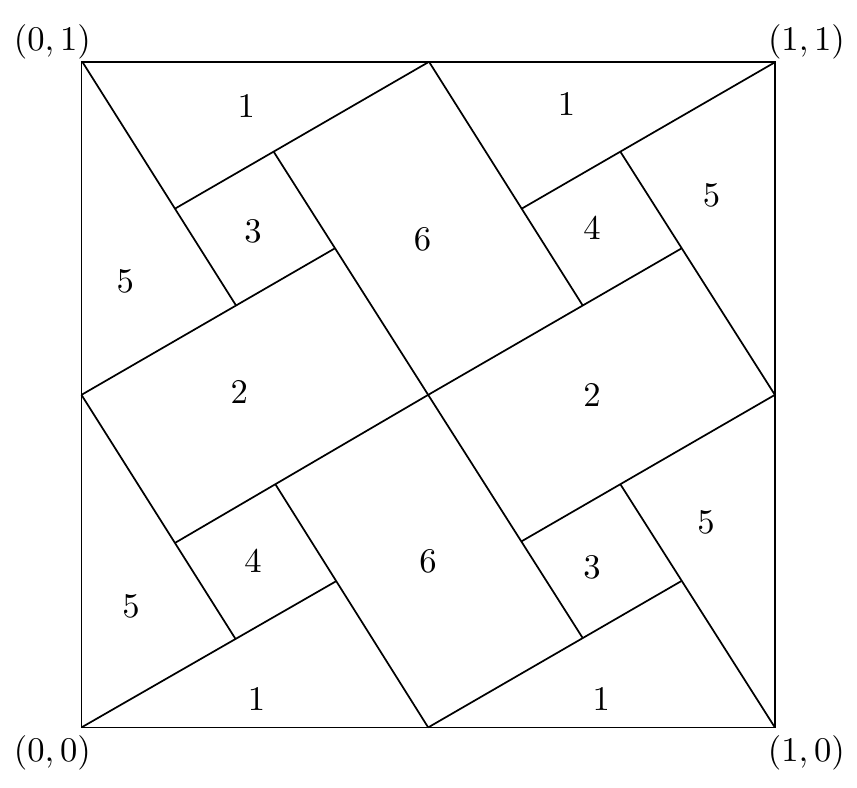}
\caption{The Markov partition for the map $\mathcal{A}$.}
\label{Maarkov}
\end{figure}
\end{centering}

The semi-conjugacy is given by the map 
\begin{equation*}
F : (\theta, \varphi) \mapsto \left( \cos 2\pi (\theta + \varphi), \cos 2\pi \theta, \cos 2\pi \varphi  \right).
\end{equation*}
A Markov partition of  
$\mathcal{A}: \mathbb{T}^2 \to \mathbb{T}^2$ when $s = 1$ is shown in Figure \ref{Maarkov}. Compare with Figure 5 from \cite{Cas}. 
For other values of $s \in \mathbb{N}$, the only difference is the slope of 
the stable and unstable manifolds.
Its image under the map $F: \mathbb{T}^{2} \to \mathbb{S}$ is a Markov partition for the pseudo-Anosov map 
$T_{s} : \mathbb{S} \to \mathbb{S}$.
Write
\begin{equation*}
I = \left\{ (t, t) \mid 0 \leq t \leq 1/2 \right\}.
\end{equation*}
Note that $F(I) \subset \ell_{0}$. The following lemma is immediate:

\begin{lem}
The push-forward of the normalized Lebesgue measure on $I$ 
under the semi-conjugacy $F$, which is a probability measure on $\ell_{0} \cap \mathbb{S}$, 
corresponds to the free density of states measure (\ref{freedensity}) under the identification (\ref{identification}).
\end{lem}

Consider the union of elements of the Markov partition of $\mathbb{T}^{2}$, 1, 2, 4, 5, and 6, as in Figure \ref{stablemanifolds}. 
Compare with Figure 3 from \cite{DG12}. 
Let us denote the image of this union of elements under $F$ by  
$\mathcal{R}_{0}$, and the continuation of $\mathcal{R}_{0}$ in $\lambda >0$ by  
$\mathcal{R}_{\lambda}$. The following statement can be proven by repeating the proof of Claim 3.2 of \cite{DG12}. 

\begin{prop}\label{projection}
Consider the measure of maximal entropy of $T_{s} |_{\Lambda_{\lambda}}$ and restrict it to 
$\mathcal{R}_{\lambda}$. Normalize 
this measure and project it to $\ell_{\lambda}$ along the stable manifolds. 
Then, the resulting probability measure on $\ell_{\lambda}$ 
corresponds to the density of states measure 
$\nu_{\lambda}$ under the identification (\ref{identification}).
\end{prop}

\begin{centering}
\begin{figure}[t]
\includegraphics[scale=1.00]{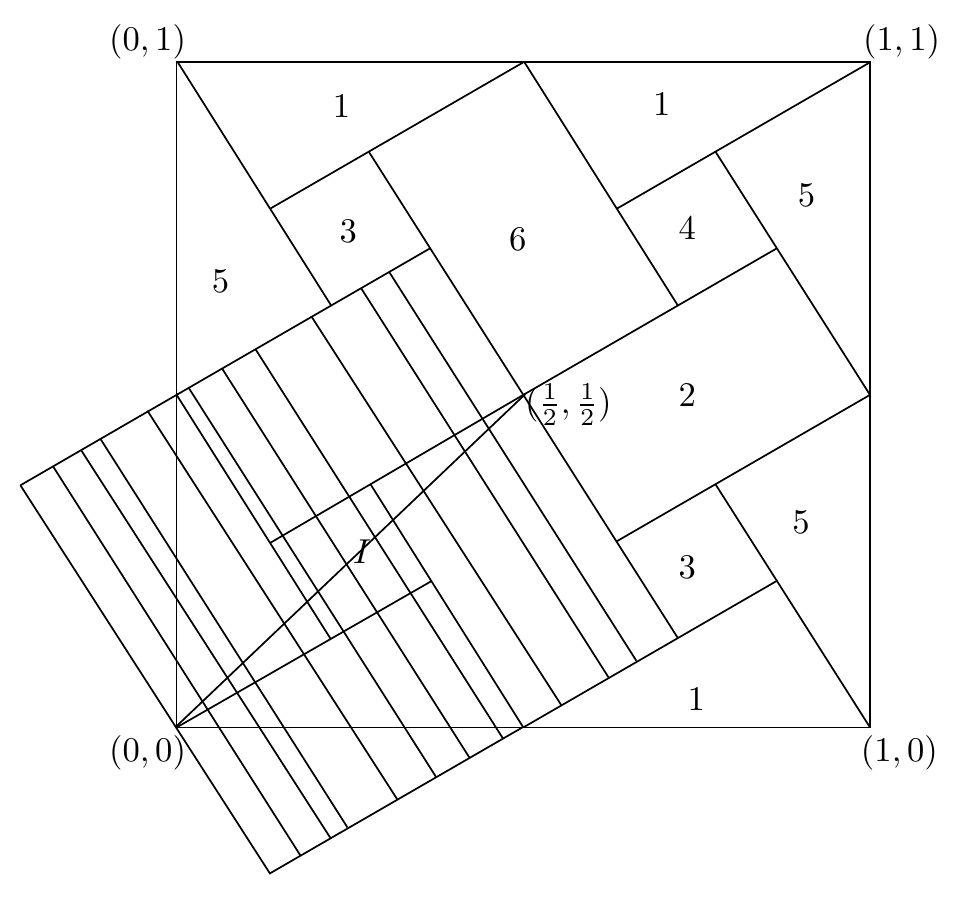}
\caption{A Markov partition for the map $\mathcal{A}$, line segment $I$, and the stable manifolds.}
\label{stablemanifolds}
\end{figure}
\end{centering}

This immediately implies the following:

\begin{thm}\label{dimension}
For every $\lambda > 0$, 
the density of states measure $\nu_{\lambda}$ is exact-dimensional. That is, for $\nu_{\lambda}$-almost 
every $E \in \mathbb{R}$, we have
\begin{equation*}
\lim_{\epsilon \downarrow 0} \frac{\log \nu_{\lambda}( E - \epsilon, E + \epsilon )}{ \log \epsilon } = d_{\lambda}, 
\end{equation*}
where $d_{\lambda}$ satisfies 
\begin{equation*}
\lim_{\lambda \downarrow 0} d_{\lambda} = 1.
\end{equation*}
\end{thm}

\begin{proof}
The first claim is an immediate consequence of Proposition \ref{projection}. 
The second claim follows verbatim from the repetition of Theorem 1.1 of \cite{DG12}. 
\end{proof}

We also have the following:

\begin{prop}[\cite{Pollicott}]\label{Lyapunov}
The stable and unstable Lyapunov exponents are analytic functions of $\lambda > 0$.
\end{prop}

For any Cantor set $K$, we denote the \emph{thickness} of $K$ by $\tau(K)$. For the definition of thickness, see 
for example, chapter 4 of \cite{PalisTakens}. By Theorem 2 of \cite{DGET} and 
by repeating the proof of Theorem 1.1 of \cite{DG11}, 
we obtain the following: 

\begin{thm}\label{thicknesstheorem}
We have 
\begin{equation*}
\lim_{\lambda \to \infty} \dim_{H} \Sigma_{\lambda} = 0, \text{  and \hspace{0.1mm} }
\lim_{\lambda \downarrow 0} \tau( \Sigma_\lambda ) = \infty.
\end{equation*}
\end{thm}

\begin{prop}\label{sym}
The density of states measure $\nu_{\lambda}$ is symmetric with respect to the origin. 
In particular, the spectrum $\Sigma_{\lambda}$ is symmetric with respect to the origin.
\end{prop}

\begin{proof}
Denote by $H_{\omega}^{(N)}$ the restriction of $H_{\omega}$ to the interval $[0, N-1]$ with Dirichlet boundary conditions.
Let $\psi$ be an eigenvector of $H_{\omega}^{(N)}$ and $E$ be the corresponding eigenvalue.
Let us define $\phi \in l^{2}([0, N-1])$ by 
\begin{equation*}
\phi(n) = (-1)^{n} \psi(n) \ \ ( n = 0, 1, \cdots ,N-1 ).
\end{equation*}
Then, since 
\begin{equation*}
\begin{aligned}
( H_{\omega}^{(N)} \phi )(n) &= \omega(n+1) \phi(n+1) + \omega(n) \phi(n-1) \\
&= (-1)^{n+1} \omega(n+1) \psi(n+1) + (-1)^{n-1} \omega(n) \psi(n-1)  \\
&= (-1)^{n+1} E \, \psi(n) \\
&= -E \, \phi(n),
\end{aligned}
\end{equation*}
$-E$ is also an eigenvalue of $H^{(N)}_{ \omega }$.
Therefore, the set of eigenvalues of $H_{ \omega }^{(N)}$ is symmetric with respect to the origin.
Therefore, for any interval $A \subset (0, \infty)$, 
\begin{equation*}
\begin{aligned}
\nu_{\lambda} \left( A \right) &= \lim_{N \to \infty} \frac{1}{N} \# \left\{  \text{eigenvalues of $H_{ \omega }^{(N)}$ that are in $A$} \right\} \\
&=  \lim_{N \to \infty} \frac{1}{N} \# \left\{  \text{eigenvalues of $H_{ \omega }^{(N)}$ that are in $(-A)$} \right\} \\
&= \nu_{\lambda}(-A).
\end{aligned}
\end{equation*}
This concludes the first claim. Since the spectrum is the topological support of the density of states measure, 
the second claim also follows. 
\end{proof}

\section{The Labyrinth Model}\label{ohohohoho}
Let $a_{i}, b_{i} > 0 \ (i = 1, 2)$ be real numbers, and let $s$ be a positive integer. 
Let $\omega_{i} \in \Omega^{(s)}_{a_{i}, b_{i}} \ (i = 1, 2)$ and $\lambda_i$ be the corresponding coupling constants. 

\subsection{The Labyrinth model}

We define the Labyrinth model. Write 
\begin{equation*}
A^{e} = \left\{ ( m, n ) \in \mathbb{Z}^2 \, | \, m + n \text{ is even}  \right\},  \text{  and }  \
A^{o} = \left\{ ( m, n ) \in \mathbb{Z}^2 \, | \, m + n \text{ is odd} \right\}.
\end{equation*}
Using $\omega_1, \omega_2$, we realign the lattices of $A^e$ and $A^o$.  
See Figure \ref{lalalabyrinth}. We denote this again by $A^e$ and $A^o$ (we use this identification freely). 
We define the operator ${\hat{H}}_{\omega_{1}, \omega_{2}}$, which acts on $l^{2}(A^e \cup A^o)$, by 
 \begin{equation*}
 \begin{aligned}
\left[ \hat{H}_{\omega_{1}, \omega_{2}} \psi \right] (m, n) &=  \omega_{1}(m+1) \omega_{2}(n+1) \psi( m+1, n+1 ) \\
&\hspace{10mm}+\omega_{1}(m+1) \omega_{2}(n) \psi( m+1, n-1 ) \\
&\hspace{14mm}+ \omega_{1}(m) \omega_{2}(n+1) \psi( m-1, n+1 ) \\
&\hspace{24mm} + \omega_{1}(m) \omega_{2}(n) \psi( m-1, n-1 ). 
\end{aligned}
 \end{equation*}
Every lattice is connected diagonally and the strength of the bond is 
equal to the product of the sides of the rectangle.  
With appropriate scaling, we can always assume that $b_i = 1 \ (i= 1, 2)$. 
We assume this scaling throughout this section.
In a similar way, we define the operators $\hat{H}^{e}_{\omega_{1}, \omega_{2}}$ and 
$\hat{H}^{o}_{\omega_{1}, \omega_{2}}$, which act on 
$l^{2}( A^{e})$ and $l^{2}(A^{o} )$ respectively.
From here, we drop the subscripts $\omega_{1}, \omega_{2}$ if no confusion can arise.
Note that 
\begin{equation*}
\hat{H} = \hat{H}^{e} \oplus \hat{H}^{o}.
\end{equation*}
It is natural to expect that the spectral properties of 
$\hat{H}^{e}$ and $\hat{H}^{o}$ are the same, and in fact, 
the spectra and the density of states measures coincide for the three operators. 
For the proof, we need the notion of \emph{Delone dynamical systems}, 
and \emph{linear repetitivity} of Delone dynamical systems. See, for example, \cite{LP}.

\begin{centering}
\begin{figure}[t]
\includegraphics[scale=1.00]{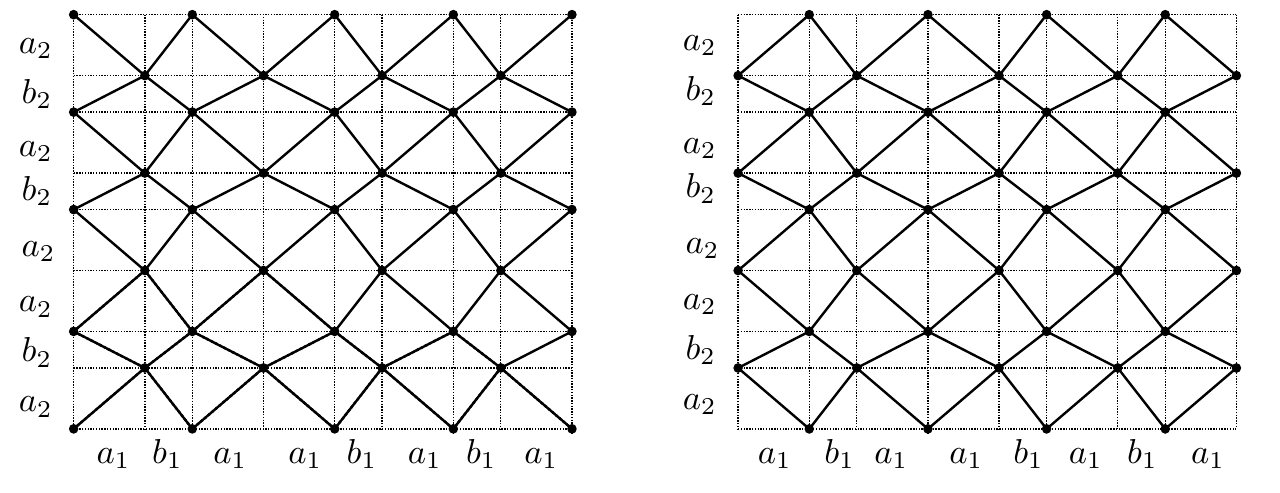}
\caption{$A^{e}$ (left) and $A^{o}$ (right).}
\label{lalalabyrinth}
\end{figure}
\end{centering}

\begin{prop}
Let us denote the density of states measures of $\hat{H}^{e}, \hat{H}^{o}$ and $\hat{H}$ by  
$\hat{\nu}^{e}, \hat{\nu}^{o}$ and $\hat{\nu}$, respectively.
Then, $\hat{\nu}^{e}, \hat{\nu}^{o}$ and $\hat{\nu}$ define the same measure. In particular,  
the spectra of $\hat{H}^{e}, \hat{H}^{o}$ and $\hat{H}$ all coincide.
\end{prop} 

\begin{proof}
By Lemma \ref{twin}, $A^{e}$ and $A^{o}$ are linearly repetitive. 
Therefore, by Theorem 6.1 of \cite{LP}, Theorem 3 of \cite{LS2} and 
Lemma \ref{twin}, we have $\hat{\nu}^{e} = \hat{\nu}^{o}$. Since $\hat{H} = \hat{H}^{e} \oplus \hat{H}^{o}$, 
we get $\hat{\nu}^{e} = \hat{\nu}^{o} = \hat{\nu}$. 
\end{proof} 

By this proposition, we will restrict our attention to $\hat{H}$ below. 

\subsection{The spectrum of the Labyrinth model}

We start by proving that the spectrum of $\hat{H}_{\omega_1, \omega_2}$ is given by the product of the spectra of off-diagonal models. 

\begin{prop}\label{prop1}
We have 
\begin{equation*}
\sigma(\hat{H}_{\omega_{1}, \omega_{2}}) = \Sigma_{\lambda_1} \cdot \Sigma_{\lambda_2}. 
\end{equation*}
In particular, the spectrum $\sigma( \hat{H}_{\omega_{1}, \omega_{2}} )$ does not depend on particular choice of 
$\omega_{1}$ and $\omega_{2}$ and only depends on the coupling constants $\lambda_1, \lambda_2$. 
\end{prop}

In the proof below, we simply write $H_{\omega_1}, H_{\omega_2}$ and $\hat{H}_{\omega_{1}, \omega_{2}}$ as 
$H_1, H_2$ and $\hat{H}$, respectively. 
\begin{proof}
Let $U$ be the unique unitary map from $l^{2}(\mathbb{Z}) \otimes l^{2}(\mathbb{Z})$ to $l^{2}({\mathbb{Z}}^{2})$ so that 
for $ \psi_{1}, \psi_{2} \in  l^{2}(\mathbb{Z})$, the elementary tensor $\psi_{1} \otimes \psi_{2}$ is mapped to the element
$\psi$ of $l^{2}({\mathbb{Z}}^{2})$ given by $\psi( m, n ) = \psi_{1}(m) \psi_{2}(n)$. We have 
\begin{equation*}
\begin{aligned}
\left[ \hat{H} U( \psi_{1} \otimes \psi_{2} ) \right](m, n) &= 
\omega_{1}(m+1) \omega_{2}(n+1) \psi_{1}(m+1)\psi_{2}(n+1) \\
&\hspace{10mm}+\omega_{1}(m+1) \omega_{2}(n) \psi_{1}(m+1)\psi_{2}(n-1) \\
&\hspace{14mm}+ \omega_{1}(m) \omega_{2}(n+1) \psi_{1}(m-1)\psi_{2}(n+1) \\
&\hspace{24mm} + \omega_{1}(m) \omega_{2}(n) \psi_{1}(m-1)\psi_{2}(n-1) \\
&= \left[ H_{1} \psi_{1} \right] (m) \left[ H_{2} \psi_{2} \right] (n) \\
&= \left[ U ( H_{1} \otimes H_{2} )( \psi_{1} \otimes \psi_{2} ) \right] ( m, n ),
\end{aligned}
\end{equation*}
for all $(m, n) \in \mathbb{Z}^{2}$. 
Therefore,  
\begin{equation*}
( U^{*} \hat{H} U ) ( \psi_{1} \otimes \psi_{2} ) = 
( H_{1} \otimes H_{2} ) ( \psi_{1} \otimes \psi_{2} ). 
\end{equation*}
Since the linear combinations of elementary tensors are dense in $l^{2}(\mathbb{Z}) \otimes l^{2}(\mathbb{Z})$, we get  
$U^{*} \hat{H} U = H_{1} \otimes H_{2}$.
Therefore, the result follows from Theorem VIII 33 of \cite{RS}.
\end{proof}

Let us denote the spectrum of $\hat{H}_{ \omega_1, \omega_2 }$ by $\hat{\Sigma}_{\lambda_1, \lambda_2}$. 
\begin{proof}[Proof of Theorem \ref{mainnnnntheorem}]
By Theorem \ref{thicknesstheorem}, we have 
\begin{equation*}
\tau( \Sigma_{\lambda_1} ), \tau( \Sigma_{\lambda_2} ) > 1 + \sqrt{2}
\end{equation*}
for sufficiently small coupling constants. Therefore, by Theorem 1.4 of \cite{YT}, 
$\Sigma_{\lambda_1} \cdot \Sigma_{\lambda_2}$ is an interval for sufficiently small $\lambda_1, \lambda_2$. 
Combining this with Proposition \ref{prop1}, $\hat{\Sigma}_{\lambda_1, \lambda_2}$ is an interval for sufficiently small coupling constants. 

By Theorem \ref{thicknesstheorem}, we get 
\begin{equation*}
\dim_{H} \Sigma_{\lambda_1} + \dim_{H} \Sigma_{\lambda_2} < 1
\end{equation*}
for sufficiently large $\lambda_1, \lambda_2$. Notice that, by the symmetry of $\Sigma_{\lambda_1}$ and $\Sigma_{\lambda_2}$, we have 
\begin{equation*}
\hat{\Sigma}_{\lambda_1, \lambda_2}^+ = \exp \left( \log \Sigma_{\lambda_1}^+ + \log \Sigma_{\lambda_2}^+ \right), 
\end{equation*}
where $A^+$ denotes $A \cap (0, \infty)$. Therefore, since 
\begin{equation*}
\dim_{H} \log \Sigma^+_{\lambda_i} = \dim_{H} \Sigma_{\lambda_i}  \ (i = 1, 2), 
\end{equation*}
by Proposition 1 of chapter 4 in \cite{PalisTakens}, $\log \Sigma_{\lambda_1}^+ + \log \Sigma_{\lambda_2}^+$ has zero Lebesgue measure.  
Hense, for sufficiently large $\lambda_1, \lambda_2$, 
$\hat{\Sigma}_{\lambda_1, \lambda_2}$ is a Cantor set of zero Lebesgure measure.
\end{proof}

\subsection{Density of states measure of the Labyrinth model}

In this section, we will prove Theorem \ref{maintheorem2}. If there is no chance of confusion, we simply write the 
density of states measures of $\hat{H}$, $H_1$ and $H_2$ as $\hat{\nu}$, $\nu_1$ and $\nu_2$, respectively. 

\begin{proof}[proof of (\ref{density1})]
The proof is essentially the repetition of the proof of Proposition A.3 of  \cite{DG13}. 
For the reader's convenience, we will repeat the argument. 

Denote by $H^{(N)}_{j} \, ( j = 1, 2 )$ the restriction of $H_{j}$ to the interval $[0, N-1]$ with Dirichlet boundary conditions.
Denote the corresponding eigenvalues and eigenvectors by 
$E^{(N)}_{j, k}, \, \phi^{(N)}_{j, k}$, where $ j = 1, 2 \text{ and }  1 \leq k \leq N$. Recall that we have
\begin{equation}\label{convergence_in_distribution}
\lim_{N \to \infty} \frac{1}{N} \# \left\{  1 \leq k \leq N \mid E^{(N)}_{j, k} \in (-\infty, E]  \right\} = \nu_{j} \left( (-\infty, E] \right)
\end{equation}
for $E \in \mathbb{R}$.

Similarly, we denote by $\hat{H}^{(N)}$ the restriction of $\hat{H}$ to $[0, N-1]^{2}$ with Dirichlet boundary conditions. 
Denote the corresponding eigenvalues and eigenvectors by 
$E^{(N)}_{k}, \phi^{(N)}_{k} \ ( 1 \leq k \leq N^{2} )$. Then, we have 
\begin{equation*}
\lim_{N \to \infty} \frac{1}{N^{2}} \# \left\{  1 \leq k \leq N^{2} \mid E^{(N)}_{k} \in (-\infty, E]  \right\} = \hat{\nu} \left( (-\infty, E] \right). 
\end{equation*}
The eigenvectors $\phi^{(N)}_{j, k}$ of $H^{(N)}_{j}$ form an orthonormal basis of $l^{2}( [ 0, N-1 ] )$. 
Thus, the associated elementary tensors 
\begin{equation}\label{eigenvector}
\phi^{(N)}_{1, k_{1}} \otimes \phi^{(N)}_{2, k_{2}}  \ \ ( 1 \leq k_{1}, k_{2} \leq N )
\end{equation}
form an orthonormal basis of $l^{2}( [ 0, N-1 ] ) \otimes l^{2} ( [ 0, N-1 ] )$, which is canonically isomorphic to 
$l^{2}( [ 0, N-1 ]^{2} )$.
Moreover, the vector in (\ref{eigenvector}) is an eigenvector of $\hat{H}^{(N)}$, corresponding to the eigenvalue 
$E^{(N)}_{1, k_{1}} \cdot E^{(N)}_{2, k_{2}}$. By counting dimensions, these eigenvalues exhaust the entire set 
$\left\{ E^{(N)}_{k} \mid 1 \leq k \leq N^{2}  \right\}$. Therefore, for $E \in \mathbb{R}$, 
\begin{equation*}
\# \left\{ 1 \leq k \leq N^{2}  \mid E^{(N)}_{k} \in (-\infty, E]   \right\} = 
\# \left\{ 1 \leq k_{1}, k_{2} \leq N  \mid E^{(N)}_{1, k_{1}}  \cdot E^{(N)}_{2, k_{2}} \in (-\infty, E]  \right\}.
\end{equation*}

Let $\nu^{(N)}_{j} \ ( j = 1, 2 )$ be the probability measures on $\mathbb{R}$  
with $\nu^{(N)}_{j}( E^{(N)}_{j, k} ) = 1/N \ ( k = 1, 2, \cdots ,N )$.
Similarly, Let $\hat{\nu}^{(N)}$ be the probability measure on $\mathbb{R}$ with
$\hat{\nu}^{(N)}(E^{(N)}_{k}) = 1/N^{2} \ ( k = 1, 2, \cdots ,N^2 )$.
Then, by the above argument, we get 
\begin{equation*}
\hat{\nu}^{(N)} \left( (-\infty, E]  \right) = \iint_{\mathbb{R}^2} \chi_{(-\infty, E]}(xy) \, 
d\nu_1^{(N)}(x) d\nu_{2}^{(N)}(y). 
\end{equation*}
By (\ref{convergence_in_distribution}), $\nu^{(N)}_{i}$ converges weakly to $\nu_{i}$ 
(see, for example, chapter 13 of \cite{AK}). 
Therefore, $\nu^{(N)}_1 \times \nu^{(N)}_2$ converges weakly to $\nu_1 \times \nu_2$. By Theorem 13.16 of 
\cite{AK}, we have 
\begin{equation*}
\lim_{N \to \infty} \iint_{\mathbb{R}^2} \chi_{(-\infty, E]}(xy) \, d\nu^{(N)}_1(x) d\nu^{(N)}_2(y)
= \iint_{\mathbb{R}^{2}} \chi_{( -\infty, E ]}(xy) \, d\nu_{1}(x) d\nu_{2}(y). 
\end{equation*}
The result follows from this. 
\end{proof}

Let us define Borel measures $\bar{\nu}_{i} \, ( i = 1, 2 )$ on $\mathbb{R}$ by  
\begin{equation*}
\bar{\nu}_{i}(A) = \nu_{i} ( e^{A} ), 
\end{equation*}
where $A \subset \mathbb{R}$ is a Borel set. 
Then, the following holds.

\begin{lem}\label{proposition}
The density of states measure of the Labyrinth model $\hat{\nu}$ is given by
\begin{equation*}
\hat{\nu}(A) = 2 \left\{ ( \bar{\nu}_{1} \ast \bar{\nu}_{2} ) ( \log A^{+}) +  
( \bar{\nu}_{1} \ast \bar{\nu}_{2} )( \log A^{-}) \right\},
\end{equation*}
where $A$ is a Borel set, and $A^{+} = A \cap ( 0, \infty )$ and $A^{-} = (-A) \cap ( 0, \infty )$.
\end{lem}

\begin{proof}
Let $A \subset (0, \infty)$ be a Borel set.
Using Fubini's Theorem and change of coordinates, we get
\begin{equation*}
\begin{aligned}
\int_{\mathbb{R}^{+}} \int_{\mathbb{R}^{+}} \chi_{A}(xy) \, d\nu_{1}(x)d\nu_{2}(y) 
&= \int_{\mathbb{R}^{+}} \left(  \int_{\mathbb{R}^+{}} \chi_{A} (xy) \, d\nu_{1}(x)  \right) d\nu_{2}(y) \\
&= \int_{\mathbb{R}^{+}} \left(  \int_{\mathbb{R}} \chi_{A} \left( e^x y \right) d\bar{\nu}_{1}(x)  \right)  d\nu_{2}(y) \\
&= \int_{\mathbb{R}} \int_{\mathbb{R}} \chi_{A}  ( e^x e^y ) \, d\bar{\nu}_{1}(x) d \bar{\nu}_{2}(y) \\
&= \iint_{\mathbb{R}^{2}} \chi_{\log A} (x+y) \, d\bar{\nu}_{1}(x) d\bar{\nu}_{2}(y) \\
&= ( \bar{\nu}_{1} \ast \bar{\nu}_{2} )(\log A).
\end{aligned}
\end{equation*}
Combining this with Proposition \ref{sym}, the result follows. 
\end{proof}
Therefore, the absolute continuity of $\hat{\nu}$ is equivalent to 
the absolute continuity of $\bar{\nu}_{1} \ast \bar{\nu}_{2}$. 

Theorem 3.2 from \cite{DG13} implies the following: 

\begin{thm}\label{thmDG12}
Let $\mathcal{J} \subset \mathbb{R}$ be an interval. 
Assume that for $\lambda \in \mathcal{J}$, $\nu_{\lambda}$ is the density of states measure of $H_{\lambda}$. 
Let $\gamma : \mathbb{R} \to \mathbb{R}$ 
be a diffeomorphism, and define a Borel measure $\mu_{\lambda}$ by 
\begin{equation*}
\mu_{\lambda}( A ) = \nu_{\lambda} ( \gamma(A) ), 
\end{equation*}
where $A \subset \mathbb{R}$ is a Borel set. 
Then, for any compactly supported exact-dimensional measure $\eta$ on $\mathbb{R}$ with
\begin{equation*}
\dim_{H} \eta + \dim_{H} \mu_{\lambda} > 1 
\end{equation*}
for all $\lambda \in \mathcal{J}$, the convolution $\eta \ast \mu_{\lambda}$ is absolutely continuous 
with respect to Lebesgue measure for almost every $\lambda \in \mathcal{J}$.  
\end{thm}

\begin{proof}[proof of theorem \ref{maintheorem2}]
By Theorem \ref{mainnnnntheorem}, $\hat{\Sigma}_{\lambda_1, \lambda_2}$ is a Cantor set of zero Lebesgue measure 
for sufficiently large coupling constants. 
Therefore, since the density of states measure $\hat{\nu}_{\lambda_1, \lambda_2}$ is supported on $\hat{\Sigma}_{\lambda_1, \lambda_2}$, 
it has to be singular continuous for sufficiently large $\lambda_1, \lambda_2$. 

By Theorem \ref{dimension},  
there exists $\lambda^{\ast} > 0$ such that 
$\dim_{H} \nu_{\lambda} > \frac{1}{2}$ for all  $\lambda \in [0, \lambda^{\ast})$.
Recall that $0 \in \Sigma_{\lambda}$. Recall also that $\Sigma_{\lambda}$ is the set of intersections between 
the stable laminations and the line $\ell_{\lambda}$. 
Therefore, since the stable laminations and $\ell_{\lambda}$ both depend smoothly on $\lambda$, we can write 
\begin{equation*}
\Sigma_{\lambda_i} \cap (0, \infty) = \bigsqcup_{n=1}^{\infty} K^{(i)}_{n}(\lambda_i) \ \ (i = 1, 2),  
\end{equation*}
where $K_{n}^{(i)} (\lambda_i)$ are Cantor sets which depend naturally on $\lambda_i$. 
Let us define Borel measures 
$\bar{\nu}_{\lambda_i}^{(n)} \ (i = 1, 2, \ n \in \mathbb{N})$ by  
\begin{equation*}
\bar{\nu}_{\lambda_i}^{(n)} (A) = \nu_{\lambda_i} |_{K_{n}^{(i)}(\lambda_i)} (e^A), 
\end{equation*}
where $A \subset \mathbb{R}$ is a Borel set. 
Then, by Theorem \ref{thmDG12}, for each $(m, n) \in \mathbb{N} \times \mathbb{N}$,  
$\bar{\nu}_{\lambda_1}^{(m)} \ast \, \bar{\nu}_{\lambda_2}^{(n)}$ is absolutely continuous for almost all $(\lambda_1, \lambda_2)$. 
This implies that 
$\bar{\nu}_{\lambda_{1}} \ast \, \bar{\nu}_{\lambda_{2}}$ is absolutely continuous for almost all $(\lambda_1, \lambda_2)$. 
\end{proof}

\section*{Acknowledgements}
The author would like to acknowledge the invaluable contributions of Anton Gorodetski and David Damanik. The author would also 
like to thank May Mai and William Yessen for their many helpful conversations, and to the anonymous referee for many helpful suggestions and remarks.

\end{document}